\documentclass[reqno]{amsart}
\usepackage{color}
\usepackage{amsmath,amsthm}
\usepackage{amsfonts}

\usepackage[all]{xy}
\SelectTips{cm}{10}

\usepackage{graphicx}

\usepackage{tikz}
\usetikzlibrary{backgrounds,circuits,circuits.ee.IEC,shapes,fit,matrix}
\usetikzlibrary{arrows,positioning}
\tikzstyle{main node} =[circle,fill=white!20,draw,font=\sffamily\Large\bfseries]
\tikzstyle{terminal}=[circle,fill=white!20,draw,font=\sffamily\Large\bfseries,color=purple,fill=none]

\definecolor{myurlcolor}{rgb}{0.5,0,0}
\definecolor{mycitecolor}{rgb}{0,0,0.8}
\definecolor{myrefcolor}{rgb}{0,0,0.8}
\usepackage[pagebackref,bookmarks=false]{hyperref}
\hypersetup{colorlinks,
linkcolor=myrefcolor,
citecolor=mycitecolor,
urlcolor=myurlcolor}

\newcommand{\R}{\mathbb{R}}

\newcommand{\Var}{\mathrm{Var}}

\newcommand{\define}[1]{{\bf \boldmath{#1}}}
\newcommand{\maps}{\colon}
\newcommand{\beq}{\begin{equation}}
\newcommand{\eeq}{\end{equation}}

\theoremstyle{plain}

\newtheorem{thm}{Theorem}

\theoremstyle{definition}

\theoremstyle{remark}

\begin{document}

%-------------------------------------------------------------------

%%%%%%%%%%%% title page stuff %%%%%%%%%%%%%%%%%%%%%%%%%%

\title[The fundamental theorem of natural selection]{The fundamental theorem\\of natural selection}

\author{John C.~Baez}
\address{Department of Mathematics\\ 
University of California\\ 
Riverside CA 92521\\
USA \\
and Centre for Quantum Technologies\\ 
National University of Singapore\\ 
Singapore 117543}

\date{\today}

\begin{abstract}
Suppose we have $n$ different types of self-replicating entity, with the population $P_i$ of 
the $i$th type changing at a rate equal to $P_i$ times the fitness $f_i$ of that type.   Suppose the fitness $f_i$ is any continuous function of all the populations $P_1, \dots, P_n$.   Let $p_i$ be the fraction of replicators that are of the $i$th type. Then $p = (p_1, \dots, p_n)$ is a time-dependent probability distribution, and we prove that its speed as measured by the Fisher information metric equals the variance in fitness.   In rough terms, this says that the speed at which information is updated through natural selection equals the variance in fitness.  This result can be seen as a modified version of Fisher's fundamental theorem of natural selection.  We compare it to Fisher's original result as interpreted by Price, Ewens and Edwards.
\end{abstract}

\maketitle

\section{Introduction}
\label{introduction}

In 1930, Fisher \cite{Fisher1930} stated his ``fundamental theorem of natural selection'' as follows:
\begin{quote}
The rate of increase in fitness of any organism at any time is equal to its genetic variance in fitness at that time.
\end{quote}
Some tried to make this statement precise as follows:
\begin{quote}
The time derivative of the mean fitness of a population equals the variance of its fitness.  
\end{quote}
But this is only true under very restrictive conditions, so a controversy was ignited.  

An interesting resolution was proposed by Price \cite{Price}, and later amplified by Ewens \cite{Ewens} and Edwards \cite{Edwards}.   We can formalize their idea as follows.  Suppose we have $n$ types of self-replicating entity, and idealize the population of the $i$th type as a positive real-valued function $P_i(t)$.   Suppose
\[   \frac{d}{dt} P_i(t) = f_i(P_1(t), \dots, P_n(t)) \, P_i(t) \]
where the fitness $f_i$ is a differentiable function of the populations of every type of replicator.
The mean fitness at time $t$ is 
\[    \overline{f}(t) = \sum_{i=1}^n p_i(t) \, f_i(P_1(t), \dots, P_n(t))  \]
where $p_i(t)$ is the fraction of replicators of the $i$th type:
\[    p_i(t) = \frac{P_i(t)}{\phantom{\Big|} \sum_{j = 1}^n P_j(t) } .\]
By the product rule, the rate of change of the mean fitness is the sum of two terms:
\[    \frac{d}{dt}  \overline{f}(t) = \sum_{i=1}^n \dot{p}_i(t) \, f_i(P_1(t), \dots, P_n(t)) 
\; + \;  \sum_{i=1}^n p_i(t) \,\frac{d}{dt} f_i(P_1(t), \dots, P_n(t)).  \]
The \emph{first} of these two terms equals the variance of the fitness at time $t$.   We give the
easy proof in Theorem \ref{thm:1}.   Unfortunately, the conceptual significance of this first term is much less clear than that of the total rate of change of mean fitness.   Ewens concluded that ``the theorem does not provide the substantial biological statement that Fisher claimed''.

But there is another way out, based on an idea Fisher himself introduced in 1922: Fisher information \cite{Fisher1922}.   Fisher information gives rise to a Riemannian metric on the space of probability distributions on a finite set, called the `Fisher information metric'---or in the context of evolutionary game theory, the `Shahshahani metric' \cite{Akin1,Akin2,Shahshahani}.   Using this metric we can define the speed at which a time-dependent probability distribution changes with time.   We call this its `Fisher speed'.   Under just the assumptions already stated, we prove in Theorem \ref{thm:2} that the Fisher speed of the probability distribution 
\[    p(t) = (p_1(t), \dots, p_n(t))  \]
is the variance of the fitness at time $t$.

As explained by Harper \cite{Harper1, Harper2}, natural selection can be thought of as a learning process, and studied using ideas from information geometry \cite{Amari}---that is, the geometry of the space of probability distributions.  As $p(t)$ changes with time, the rate at which information is updated is closely connected to its Fisher speed.   Thus, our revised version of the fundamental theorem of natural selection can be loosely stated as follows:
\begin{quote}
As a population changes with time, the rate at which information is updated equals the variance of fitness. 
\end{quote}
The precise statement, with all the hypotheses, is in Theorem \ref{thm:2}.  But one lesson is this: variance in fitness may not cause `progress' in the sense of increased mean fitness, but it does cause change.

\section{The time derivative of mean fitness}
\label{sec:fitness}

Suppose we have $n$ different types of entity, which we call \define{replicators}.   Let $P_i(t),$ or $P_i$ for short, be the population of the $i$th type of replicator at time $t$, which we idealize as taking positive real values.  Then a very general form of the \define{Lotka--Volterra equations} says that
\begin{equation}
\label{eq:Lotka-Volterra}
\displaystyle{ \frac{d P_i}{d t} = f_i(P_1, \dots, P_n) \, P_i }. 
\end{equation}
where $f_i \maps [0,\infty)^n \to \R$ is  the \define{fitness function} of the $i$th type
of replicator.   One might also consider fitness functions with explicit time dependence, but
we do not do so here.   

Let $p_i(t)$, or $p_i$ for short, be the probability at time $t$ that a randomly chosen replicator will be of the $i$th type.   More precisely, this is the fraction of replicators of the $i$th type:
\begin{equation}
\label{eq:probability}
\displaystyle{  p_i = \frac{P_i}{\sum_j P_j} }. 
\end{equation}
Using these probabilities we can define the \define{mean fitness} $\overline{f}$ by
\begin{equation}
\label{eq:mean}
 \overline{f} = \sum_{j = 1}^n p_j \, f_j(P_1, \dots, P_n)   
\end{equation}
and the \define{variance in fitness} by
\begin{equation}
\label{eq:variance}
   \Var(f) = \sum_{j = 1}^n p_j \left(f_j(P_1, \dots, P_n) - \overline{f}\right)^2.  
\end{equation}
These quantities are also functions of $t$, but we suppress the $t$ dependence in our notation.

Fisher said that the variance in fitness equals the rate of change of mean fitness.  Price \cite{Price}, Ewens \cite{Ewens} and Edwards \cite{Edwards} argued that Fisher only meant to equate \emph{part} of the rate of change in mean fitness to the variance in fitness.   We can see this in the present context as follows.  The time derivative of the mean fitness is the sum of two terms:
\begin{equation}
\label{eq:fitness_change}
\frac{d\overline{f}}{dt} = \sum_{i=1}^n \dot{p}_i \, f_i(P_1(t), \dots, P_n(t))  \; + \; \sum_{i=1}^n p_i \, \frac{d}{dt} f_i(P_1(t), \dots, P_n(t)) 
\end{equation}
and as we now show, the \emph{first} term equals the variance in fitness.  

\begin{thm}
\label{thm:1}
Suppose positive real-valued functions $P_i(t)$ obey the Lotka--Volterra equations for some continuous functions $f_i \maps [0,\infty)^n \to \R$.   Then
\[  \sum_{i=1}^n \dot{p}_i \, f_i(P_1(t), \dots, P_n(t)) = \Var(f) .\]
\end{thm}

\begin{proof}
First we recall a standard formula for the time derivative $\dot{p}_i$.  Using the definition of $p_i$ in equation \eqref{eq:probability}, the quotient rule gives
\[ 
\displaystyle{ \dot{p}_i = \frac{\dot{P}_i}{\sum_j P_j} 
 - \frac{P_i \left(\sum_j \dot{P}_j \right)}{(  \sum_j P_j )^2 } }
\]
where all sums are from $1$ to $n$.    Using the Lotka--Volterra equations this becomes
\[ 
\displaystyle{ \dot{p}_i = \frac{f_i P_i}{\sum_j P_j} 
 - \frac{P_i \left(\sum_j f_j P_j \right)}{(  \sum_j P_j )^2 } }
\]
where we write $f_i$ to mean $f_i(P_1, \dots, P_n)$, and similarly for $f_j$.
Using the definition of $p_i$ again, this simplifies to:
\[ 
\dot{p}_i = f_i p_i  - \big( \sum_j f_j p_j \big) p_i 
\]
and thanks to the definition of mean fitness in equation \eqref{eq:mean}, this reduces to
the well-known \define{replicator equation}:
\begin{equation}
\label{eq:replicator}
\dot{p}_i  = \left(f_i - \overline{f}\right) p_i .
\end{equation}

Now, the replicator equation implies
\begin{equation}
\label{eq:rate1}
\sum_i f_i \dot{p}_i  = \sum_i f_i \left(f_i - \overline{f}\right) p_i .
\end{equation}
On the other hand,
\begin{equation}
\label{eq:rate2}
  \sum_i \overline{f} (f_i - \overline{f}) p_i = \overline{f} \sum_i (f_i - \overline{f}) p_i = 0 
\end{equation}
since $\sum_i f_i p_i = \overline{f}$ but also $\sum_i \overline{f} p_i = \overline{f}$.  Subtracting
equation \eqref{eq:rate2} from equation \eqref{eq:rate1} we obtain
\[   \sum_i f_i \dot{p}_i = \sum_i (f_i - \overline{f}) (f_i - \overline{f}) p_i  \]
or simply
\[    \sum_i f_i \dot{p}_i = \Var(f) .   \qedhere\]
\end{proof}

The second term of equation (\ref{eq:fitness_change}) only vanishes in special cases, e.g.\ when the fitness functions $f_i$ are constant.   When the second term vanishes we have 
\[   \frac{d\overline{f}}{dt}  = \Var(f) .\]
This is a satisfying result.  It says the mean fitness does not decrease, and it increases whenever some replicators are more fit than others, at a rate equal to the variance in fitness.   But we would like a more general result, and we can state one using a concept from information theory: the Fisher speed.

\section{The Fisher speed}
\label{sec:Fisher}

While Theorem \ref{thm:1} allows us to express the variance in fitness in terms of the time derivatives of the probabilities $p_i$, it does so in a way that also explicitly involves the fitness functions $f_i$.   We now prove a simpler formula for the variance in fitness, which equates it with the square of the `Fisher speed' of the probability distribution $p = (p_1, \dots, p_n)$.

The space of probability distributions on the set $\{1, \dots, n\}$ is the $(n-1)$-simplex
\[ \Delta^{n-1} = \{ (x_1, \dots, x_n) : \; x_i \ge 0, \; \sum_{i=1}^n x_i = 1  \} \]
The \define{Fisher metric} is the Riemannian metric $g$ on the interior of the $(n-1)$-simplex such that given a point $p$ in the interior of $\Delta^{n-1}$ and two tangent vectors $v,w$ we have
\[ g(v,w) =  \sum_{i=1}^n \frac{v_i w_i}{p_i} . \]
Here we are describing the tangent vectors $v,w$ as vectors in $ \R^n$ with the property that the sum of their components is zero: this makes them tangent to the $(n-1)$-simplex.  We are demanding that $x$ be in the interior of the simplex to avoid dividing by zero, since on the boundary of the simplex we have $p_i = 0$ for at least one choice of $i$.

If we have a time-dependent probability distribution $p(t)$ moving in the interior of the $(n-1)$-simplex as a function of time, its \define{Fisher speed} is defined by
\[ \sqrt{g(\dot{p}(t), \dot{p}(t))} = \left( \sum_{i=1}^n \frac{\dot{p}_i(t)^2}{p_i(t)} \right)^{\!\! 1/2} \]
if the derivative $ \dot{p}(t)$ exists.   This is the usual formula for the speed of a curve moving in a Riemannian manifold, specialized to the case at hand.

These are all the formulas needed to prove our result.  But for readers unfamiliar with the Fisher metric, a few words may provide some intuition.  The factor of $1/p_i$ in the Fisher metric changes the geometry of the simplex so that it becomes round, with the geometry of a portion of a sphere in $\R^n$.  But more relevant here is the Fisher metric's connection to relative information---a generalization of Shannon information that depends on two probability distributions rather than just one \cite{CoverThomas}.  Given probability distributions $p, q \in \Delta^{n-1}$, the \define{information of} $q$ \define{relative to} $p$ is
\[  I(q,p) = \sum_{i = 1}^n q_i \ln\left(\frac{q_i}{p_i}\right).  \]
This is the amount of information that has been updated if one replaces the prior distribution $p$ with the posterior $q$.   So, sometimes relative information is called the `information gain'.  It is also called `relative entropy' or `Kullback--Leibler divergence'.   It has many applications to biology \cite{BaezPollard,Harper1,Harper2,Leinster}.

Suppose $p(t)$ is a smooth curve in the interior of the $(n-1)$-simplex.  We can ask the rate at which information is being updated as time passes.  Perhaps surprisingly, an easy calculation gives
\[  \frac{d}{dt} I(p(t), p(t_0))\Big|_{t = t_0} = 0. \]
Thus, to first order, information is not being updated at all at any time $t_0 \in \R.$   However, another well-known calculation (see e.g.\ \cite{Baez}) shows that
\[  \frac{d^2}{dt^2} I(p(t), p(t_0))\Big|_{t = t_0} =  g(\dot{p}(t_0), \dot{p}(t_0))  . \]
So, to second order in $t-t_0$, the square of the Fisher speed determines how much information is updated when we pass from $p(t_0)$ to $p(t)$.

\begin{thm}
\label{thm:2}
Suppose positive real-valued functions $P_i(t)$ obey the Lotka--Volterra equations for some continuous functions $f_i \maps [0,\infty)^n \to \R$.   Then the square of the Fisher speed of the probability distribution $p(t)$ is the variance of the fitness:
\[  g(\dot{p}, \dot{p})  = \Var(f(P)) .\]
\end{thm}

\begin{proof}
Consider the square of the Fisher speed
\[  g(\dot{p}, \dot{p}) = \sum_{i=1}^n \frac{\dot{p}_i^2}{p_i}  \]
and use the replicator equation 
\[   \dot{p}_i  = \left(f_i - \overline{f}\right) p_i \]
obtaining
\[
\begin{array}{ccl} 
\displaystyle{ g(\dot{p}, \dot{p})} &=&
\displaystyle{ \sum_{i=1}^n ( f_i(P) - \overline f(P))^2 p_i } \\ \\
&=& \mathrm{Var}(f)
\end{array} 
\]
as desired. 
\end{proof}

The generality of this result is remarkable.   Formally, \emph{any} autonomous system of first-order differential equations
\[    \frac{d}{dt} P_i(t) = F_i(P_1(t), \dots, P_n(t))  \]
can be rewritten as Lotka--Volterra equations
\[     \frac{d}{dt} P_i(t) = f_i(P_1(t), \dots, P_n(t)) \, P_i(t)   \]
simply by setting
\[    f_i(P_1, \dots, P_n) = F_i(P_1 , \dots , P_n) / P_i  .\]
In general $f_i$ is undefined when $P_i = 0$, but this not a problem if we restrict ourselves to situations where all the populations $P_i$ are positive; in these situations Theorems \ref{thm:1} and \ref{thm:2} apply.

\subsection*{Acknowledgments}

This work was supported by the Topos Institute.  I thank Marc Harper for his invaluable continued help with this subject, and evolutionary game theory more generally.  I also thank Rob Spekkens for some helpful comments.


\begin{thebibliography}{99}

\bibitem{Akin1} E.\ Akin, \textsl{The Geometry of Population Genetics}, Springer, Berlin, 1979.

\bibitem{Akin2} E.\ Akin, The differential geometry of population genetics and evolutionary games, in \textsl{Mathematical and Statistical Developments of Evolutionary Theory}, ed.\  S.\ Lessard, Springer, Berlin, 1990, pp.\ 1--93.

\bibitem{Amari} S.\ Amari, \textsl{Information Geometry and its Applications}, Springer, Berlin, 2016.

\bibitem{Baez} J.\ C.\ Baez, Information geometry, Part 7, 2011.  Available at \href{https://math.ucr.edu/home/baez/information/}{https://math.ucr.edu/home/} \href{https://math.ucr.edu/home/baez/information/}{baez/information/}.

\bibitem{BaezPollard} J.\ C.\ Baez and B.\ S.\ Pollard, Relative entropy in biological systems, \textsl{Entropy} \textbf{18} (2016), 46.  Available as \href{https://arxiv.org/abs/1512.02742}{arXiv:1512.02742}.

\bibitem{CoverThomas} T.\ M.\ Cover and J.\ A.\ Thomas, \textsl{Elements of Information Theory}, 2nd Edition, Wiley-Interscience, New York, 2006.

\bibitem{Edwards} A.\ W.\ F.\ Edwards, The fundamental theorem of natural selection, 
\textsl{Biol.\ Rev.\ }\textbf{69} (1994), 443--474.

%\bibitem{EthierKurtz} S.\ N.\ Ethier and T.\ G.\ Kurtz,  \textsl{Markov Processes: Characterization and Convergence}, Wiley-Interscience, New York, 2005.

\bibitem{Ewens} W.\ J.\ Ewens, An interpretation and proof of the Fundamental Theorem of Natural Selection, \textsl{Theor.\ Popul.\ Biol.} \textbf{36} (1989), 167--180.

\bibitem{Fisher1922} R.\ A.\ Fisher, On the mathematical foundations of theoretical statistics, \textsl{Philos.\ Trans.\ A Math.\ Phys.\ Eng.\ Sci.} \textbf{222} (1922), 309--368.

\bibitem{Fisher1930} R.\ A.\ Fisher, \textsl{The Genetical Theory of Natural Selection}, 
Clarendon Press, Oxford, 1930.

\bibitem{Harper1} M.\ Harper, Information geometry and evolutionary game theory, 2009.  Available as \href{http://arxiv.org/abs/0911.1383}{arXiv:0911.1383}.

\bibitem{Harper2} M.\ Harper, The replicator equation as an inference dynamic, 2009.  Available as \href{http://arxiv.org/abs/0911.1763}{arXiv:0911.1763}.

%\bibitem{Hobson} A.\ Hobson, \textsl{Concepts in Statistical Mechanics}, Gordon and Breach, New York, 1971.

%\bibitem{Mitchell} M.\ Mitchell, {\em An Introduction to Genetic Algorithms}, MIT Press, Cambridge, 1998.

%\bibitem{Nielsen} R.\ Nielsen, ed., \textsl{Statistical Methods in Molecular Evolution}, Springer, Berlin, 2005.

\bibitem{Leinster} T.\ Leinster, \textsl{Entropy and Diversity: The Axiomatic Approach}, 
Cambridge U.\ Press, Cambridge, 2021.  Available as \href{https://arxiv.org/abs/2012.02113}{arXiv:2012.02113}.

\bibitem{Price} G.\ R.\ Price,  Fisher's ``fundamental theorem" made clear, \textsl{Annals of Human Genetics} \textbf{36} (1972), 129--140. 

%\bibitem{SoberSteel} E.\ Sober and M.\ Steel, Entropy increase and information loss in Markov models of evolution, \textsl{Biol.\ Philos.} \textbf{26} (2011), 223--250.

\bibitem{Shahshahani} S.\ Shahshahani, A new mathematical framework for the study of linkage and selection, \textsl{Mem.\ Am.\ Math.\ Soc.} \textbf{211}, 1979.






\end{thebibliography}
\end{document}